\documentclass[11pt]{article}

\pdfoutput=1

\usepackage{amsmath,amsthm,amssymb}
\usepackage{parskip}
\usepackage[margin=1in]{geometry}
\usepackage{mathptmx}
\usepackage[round]{natbib}
\usepackage{booktabs}
\usepackage{url}
\usepackage{fancyvrb}
\usepackage{graphicx}
\usepackage[labelfont=bf,labelsep=period,justification=raggedright]{caption}
\usepackage[labelfont=bf,justification=raggedright]{subcaption}

\newtheorem{theorem}{Theorem}
\newtheorem{lemma}[theorem]{Lemma}
\newtheorem{proposition}[theorem]{Proposition}
\newtheorem{corollary}[theorem]{Corollary}

\begin{document}

\begin{flushleft}
{\LARGE \bf
{Deciphering Interactions in Causal Networks without}

{Parametric Assumptions}
}

\bigskip

Yang Zhang$^1$, Mingzhou Song$^{1*}$

\bigskip

{\bf{1}} Department of Computer Science, New Mexico State University, Las Cruces, NM 88003, USA

* To whom correspondence may be addressed: joemsong@cs.nmsu.edu

\end{flushleft}

\bigskip

\section*{Abstract}

With the assumption that the effect is a mathematical function of the cause in a causal relationship, \textsc{FunChisq}, a chi-square test defined on a non-parametric representation of interactions, infers network topology considering both interaction directionality and nonlinearity.  Here we show that both experimental and in silico biological network data suggest the importance of directionality as evidence for causality.  Counter-intuitively, patterns in those interactions effectively revealed by \textsc{FunChisq} enlist an experimental design principle essential to network inference -- perturbations to a biological system shall make it transits between linear and nonlinear working zones, instead of operating only in a linear working zone.

\section{Introduction}

We contemplate that nonlinear dynamics can be exploited to improve identifiability of the network dependency structure among random variables in a largely unknown system.  This is in contrast to the common belief that linear dynamics is more revealing about the underlying system than nonlinear dynamics.  This perspective seems to be particularly relevant to study biological systems where measurements are sparse and noisy but dynamics can be highly nonlinear.  We explore this opportunity for network inference from observed dynamic or perturbed data of a biological system, by promoting nonlinear functional relationships among variables in the system.  Our problem is to determine the statistical strength of non-constant function $f: X \to Y$ from the observed data without parametric assumptions about $f$.

Although computational methods seeking evidence to support inference of causal biological networks from large-scale omics data have been an active pursuit over the past two decades, the performance on real biological systems has remained poor and is sometimes not much better than random guessing \citep{Marbach2012}.  This status seems to be a consequence of non-ideal interplay between experimental design and network inference methods.  Functional dependencies among random variables can suggest causality from the independent variable to the dependent variable.  Continuous linear, switch-like, sigmoidal functions, and copula \citep{Kim2008} have been used in regression analysis.  However, due to both biological complexity and variability in experimental data, the usage of specific parametric functions can often be difficult to justify in advance for not well-understood biological systems.  Testing for many possible parametric forms can be either computationally inefficient or statistically ineffective.  To overcome such challenges, we collect evidence for causality of a potential interaction from its contingency table, which is a non-parametric representation capable of qualitatively approximating any function to a sufficient accuracy under strong data uncertainty -- prevalent in biological experiments.  Although many statistical measures are applicable to contingency tables, including mutual information / G-test, various correlation coefficients, joint or conditional likelihood, Bayesian information criterion, and Pearson's chi-square test, they are designed for associations and not generally sensitive to functional relationships.  

We have previously developed generalized logical network inference using the classic Pearson's chi-square test of independence, which can detect causal nonlinear interactions when observed temporal data are available \citep{Song2009}.  Extending our prior work, here we present a method called \textsc{FunChisq} to infer network topology by detecting functional dependencies among random variables aiming at interaction directionality and nonlinearity.  Empirical evidence for the effectiveness of \textsc{FunChisq} is gained from the DREAM5 Challenges that assess reverse engineering methods for biological networks.  

The accumulating evidence on the effectiveness of \textsc{FunChisq} reveals an interesting systems experimental design principle -- to sample linear-to-nonlinear transition with sufficient detail in addition to the normal state and an extremely perturbed state, because such sampling can allow nonlinear functional relationships to be employed for insights impossible with linear ones.  For example, in addition to the wild type and a homozygous gene mutant, it is desirable to include a heterozygous mutant in the experimental design for a functional test approach to catch the linear-to-nonlinear transition dynamics that may reveal causality.

\section{Methods}

\subsection{The non-constant functional chi-square test}

Our goal is to test from an observed contingency table whether discrete random variable $Y$ is a non-constant function $f$ of discrete random variable $X$.  We represent the potential function $f$ using a contingency table, where we let $X$ be the row variable and $Y$ be the column variable representing the potential cause and effect, respectively.  $X$ is formed by a combination of multiple discrete variables which we call parents; we also call $Y$ the child.  We decide if $Y=f(X)$ is statistically supported by the data for some non-constant function $f$.

We represent an observed $r\times s$ contingency table as matrix $[n_{ij}]$ $(n_{ij}\ge 0)$.  Let $n_{i\cdot}$ be the sum of observations in row $i$ and $n_{\cdot j}$ sum of observations in column $j$.  Let $n$ be the total number of observations.  We define the functional chi-square statistic by
\begin{align}
\chi^2(f:X\to Y) & = \sum_i \chi^2(Y|X=i) - \chi^2(Y) \\
 & = \sum_{i=1}^r \sum_{j=1}^s \frac{(n_{ij}-n_{i\cdot}/{s})^2}{{n_{i\cdot}}/{s}} - \sum_{j=1}^s \frac{(n_{\cdot j}-n/s)^2}{n/s} 
 \label{eq:FCdef}
\end{align}
Each row chi-square $\chi^2(Y|X=i)$ represents deviation of $Y$ from a uniform distribution conditioned on the given $X=i$.  The child chi-square $\chi^2(Y)$ represents the deviation of $Y$ from a uniform distribution not contributed by $X$.  The difference is the deviation of $Y$ from a uniform distribution explainable by $X$.

\subsection{Properties}

We summarize the following properties for the functional chi-square test defined above:
\begin{description}

\item[Zeros:] $\chi^2(f:X\to Y)$ is zero if the empirical joint distribution of $X$ and $Y$ can be factorized as $\hat{P}(X,Y)=\hat{P}(X)\hat{P}(Y)$, or $X$ and $Y$ are empirically statistically independent $\hat{P}(Y|X)=\hat{P}(Y)$.  This is proved as Proposition~\ref{prop:zeros} in Appendix~\ref{AppZeros}.

Both a constant function $Y=f(X)=c$ and a conditional uniform distribution $P(Y|X)=\frac{1}{s}$ are zeros of the statistic, which are desirable as neither provides evidence for causality.  

These zeros of functional chi-square are also zeros of the Pearson's chi-square statistic.

\item[Non-negativity:] The functional chi-square is non-negative for any given contingency table, justified by Corollary~\ref{cor:nn} in Appendix~\ref{AppNull}, mathematically true for any sample size (including asymptotically).

\item[Asymmetry:] The functional chi-square test is asymmetric in $X$ and $Y$, i.e., $\chi^2(f:X\to Y) \neq \chi^2(f:Y \to X)$.  It thus does not give the same test statistics if we rotate the row and column into another contingency matrix.  This is demonstrated in the example in Fig.~\ref{fig:ex}.

\item[Asymptotics:] Under the null hypothesis of $Y$ being uniformly distributed conditioned on $X$, the functional chi-square statistic asymptotically follows a chi-square distribution with $(s-1)(r-1)$ degrees of freedom.  This result is given as Theorem~\ref{funchisq} in Appendix~{\ref{AppNull}}.

\item[Optimality:] A contingency table with a given column marginal distribution maximizes $\chi^2(f:X\to Y)$ if and only if column variable $Y$ is a function of row variable $X$ when such a contingency table is feasible.  This is established as Theorem~\ref{Opt} in Appendix~\ref{AppOpt}.

\end{description}

\subsection{An example}

The example in Fig.~\ref{fig:ex} demonstrates that the functional chi-square statistic promotes functional relationships but demotes otherwise.
\begin{figure}[htbp]
\centering
\textbf{\textsf{a}} \;
\scalebox{0.9}{
\begin{tabular}{cccc}
\toprule
     	& \#$(Y=1)$ & \#$(Y=2)$ & \#$(Y=3)$\\
\midrule
$X=1$ &  5 & 1 & 1 \\
$X=2$ &  1 & 5 & 0 \\
$X=3$ &  5 & 1 & 1 \\
\midrule
\multicolumn{4}{c}{$\chi^2(f:X\to Y) = 10.04$, df=4, $p$-value = $0.040$}\\
\bottomrule
\end{tabular}
}
\hspace{.125in}
\textbf{\textsf{b}} \;
\scalebox{0.9}{
\begin{tabular}{cccc}
\toprule
     	& \#$(X=1)$ & \#$(X=2)$ & \#$(X=3)$ \\
\midrule
$Y=1$ &  5 & 1 & 5 \\
$Y=2$ &  1 & 5 & 1 \\
$Y=3$ &  1 & 0 & 1 \\
\midrule
\multicolumn{4}{c}{$\chi^2(f:Y\to X) = 8.38$, df=4, $p$-value = $0.079$}\\
\bottomrule
\end{tabular}
}
\caption{{\bf Functional chi-squares are more sensitive to functional relationships than Pearson's chi-squares.}  The two tables are identical except with switched columns and rows.  They have an equal Pearson's chi-square statistic of 8.87 with 4 degrees of freedom and a $p$-value of 0.064, but very different functional chi-squares.  \textbf{\textsf{(a)}} A strong $X\to Y$ functional relationship has a high functional chi-square. \textbf{\textsf{(b)}} A weak $Y\to X$ functional relationship has a low functional chi-square.}
\label{fig:ex}
\end{figure}

\section{Results}

\subsection{Evaluation of \textsc{FunChisq} on DREAM5 Challenges in silico data set}

We first evaluated \textsc{FunChisq} on the in silico data set used in the network inference challenge of the DREAM5 Challenges \citep{Marbach2012}.  This data set was generated with a computer model of yeast transcription regulation and the groundtruth network is unambiguously defined.  We used data from the subnetwork that includes only the given 195 transcription factors.  This setup focuses us to evaluate how an algorithm handles bi-directionality.  The maximum number of quantization levels was set to 9 and we only evaluated pair-wise interactions as the experiments did not involve combinatorial perturbations.  Then we compared the performance of \textsc{FunChisq} with Pearson's chi-square, Pearson's correlation, and ANOVA.  Figure~\ref{fig:DREAM5insilico} shows the ROC curve and the precision-recall curve for each of the method along with the area under the curve.  \textsc{FunChisq} was a clear winner and achieved statistically significantly higher areas under both curves.
\begin{figure}[htbp]
\centering
\textbf{\textsf{a}} \includegraphics[width=.45\linewidth]{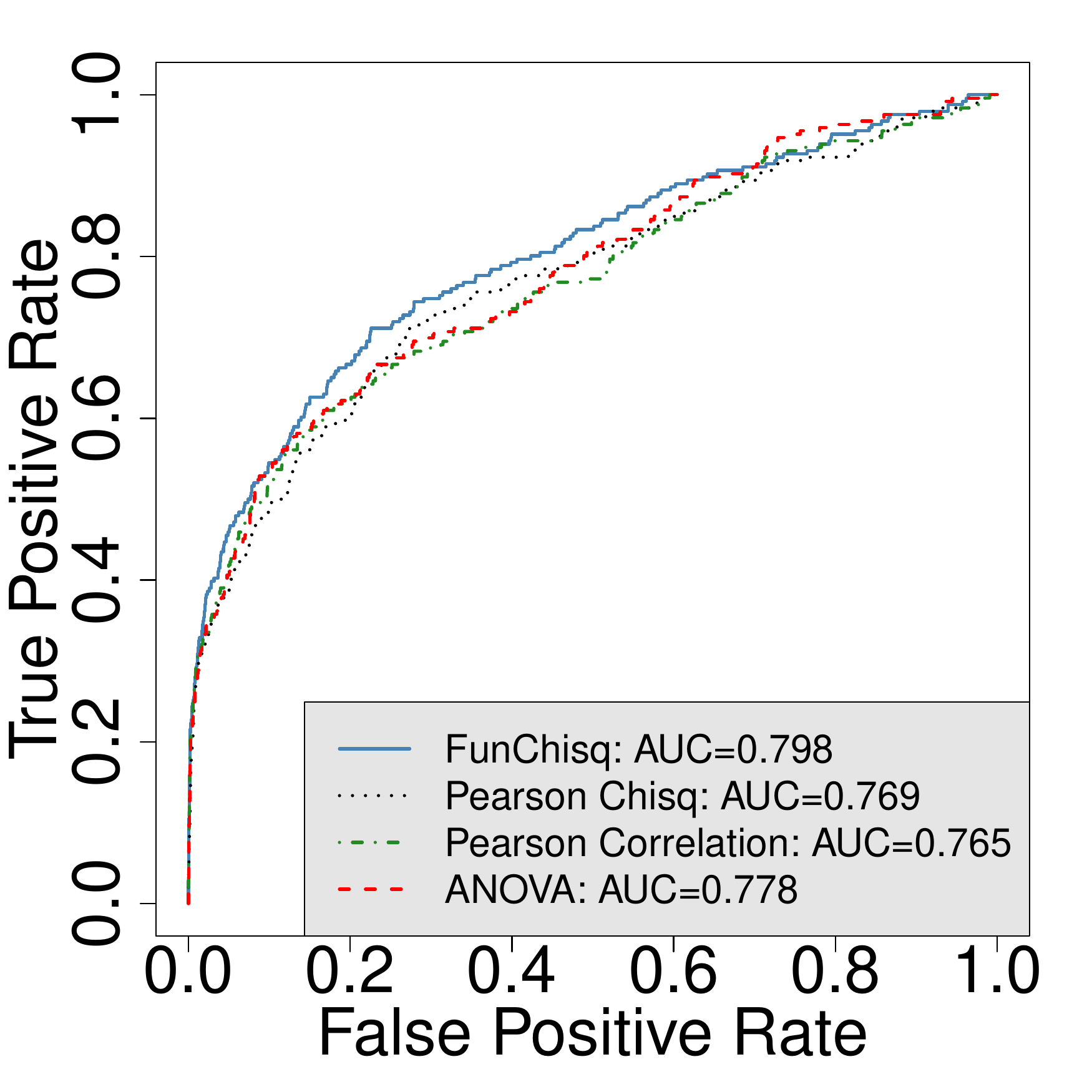}
\hspace{20pt}
\textbf{\textsf{b}} \includegraphics[width=.45\linewidth]{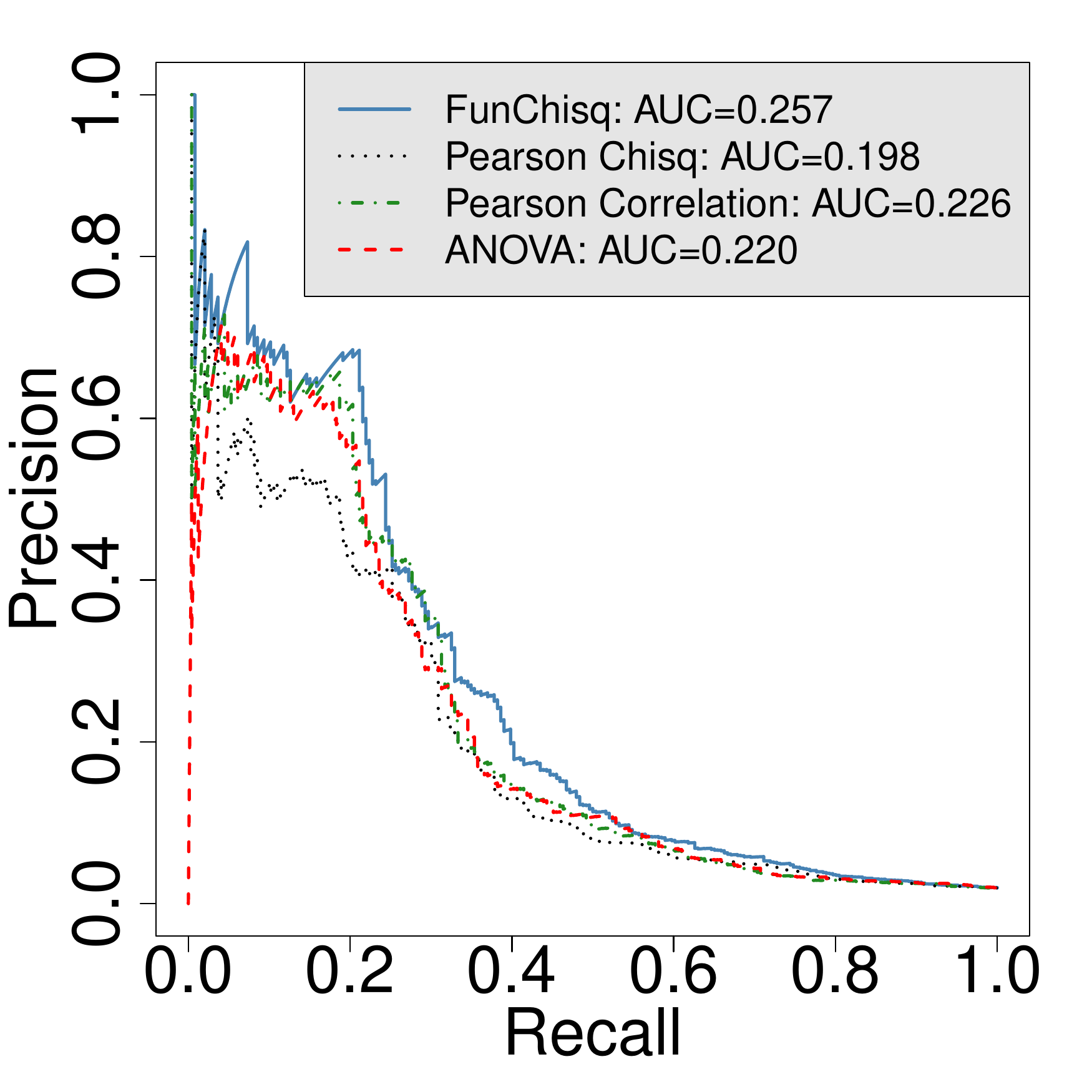}
  \caption{{\bf Advantage of functional tests on DREAM5 in silico network data.}  We compared \textsc{FunChisq} with three other methods including ANOVA (functional),  Pearson's chi-square test (non-functional), and Pearson's correlation test (non-functional).  The areas under the curve (AUCs) are given in the legend.  \textbf{\textsf{(a)}} The receiver operating curves.  The AUC of \textsc{FunChisq} is statistically significantly greater than the other three methods ($t$-test $p$-value=0.0019).  Additionally, the functional tests (\textsc{FunChisq} and ANOVA) performed better than the non-functional tests.  \textbf{\textsf{(b)}} The precision-recall curves.  AUC of \textsc{FunChisq} still performed the best, statistically significantly ($t$-test $p$-value=0.038).  }   
\label{fig:DREAM5insilico}
\end{figure}
Figure~\ref{fig:D5N1} highlights the types of interaction whose directionality is strongly discriminated by \textsc{FunChisq} from the DREAM5 in silico data set.  Specifically, there is a notable difference between functional chi-square values when the two genes switch position in each interaction, such that the direction that is more functional is promoted.
\begin{figure}[htbp]
\centering
\textbf{\textsf{a}} \; \includegraphics[width=.45\linewidth]{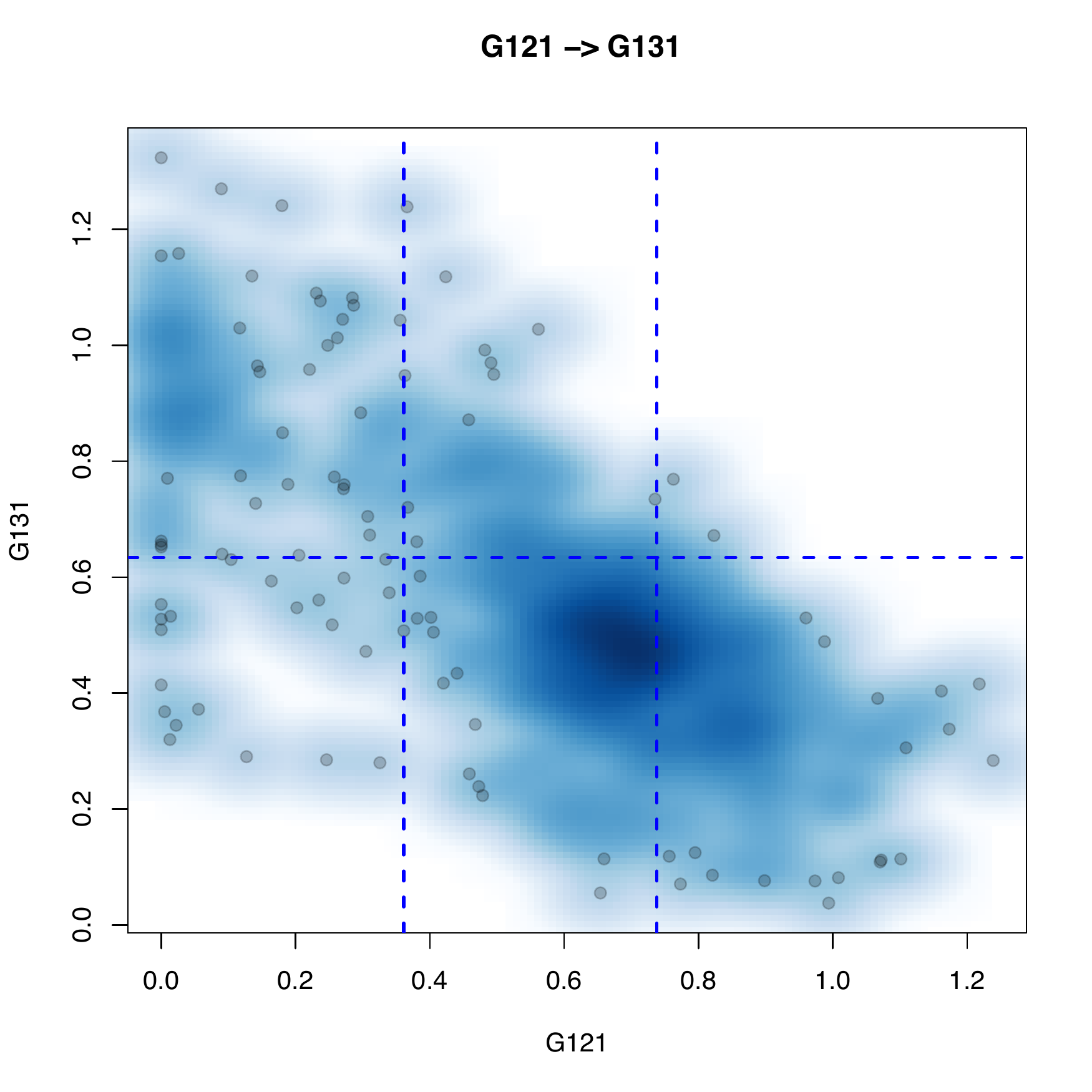}
\textbf{\textsf{b}} \; \includegraphics[width=.45\linewidth]{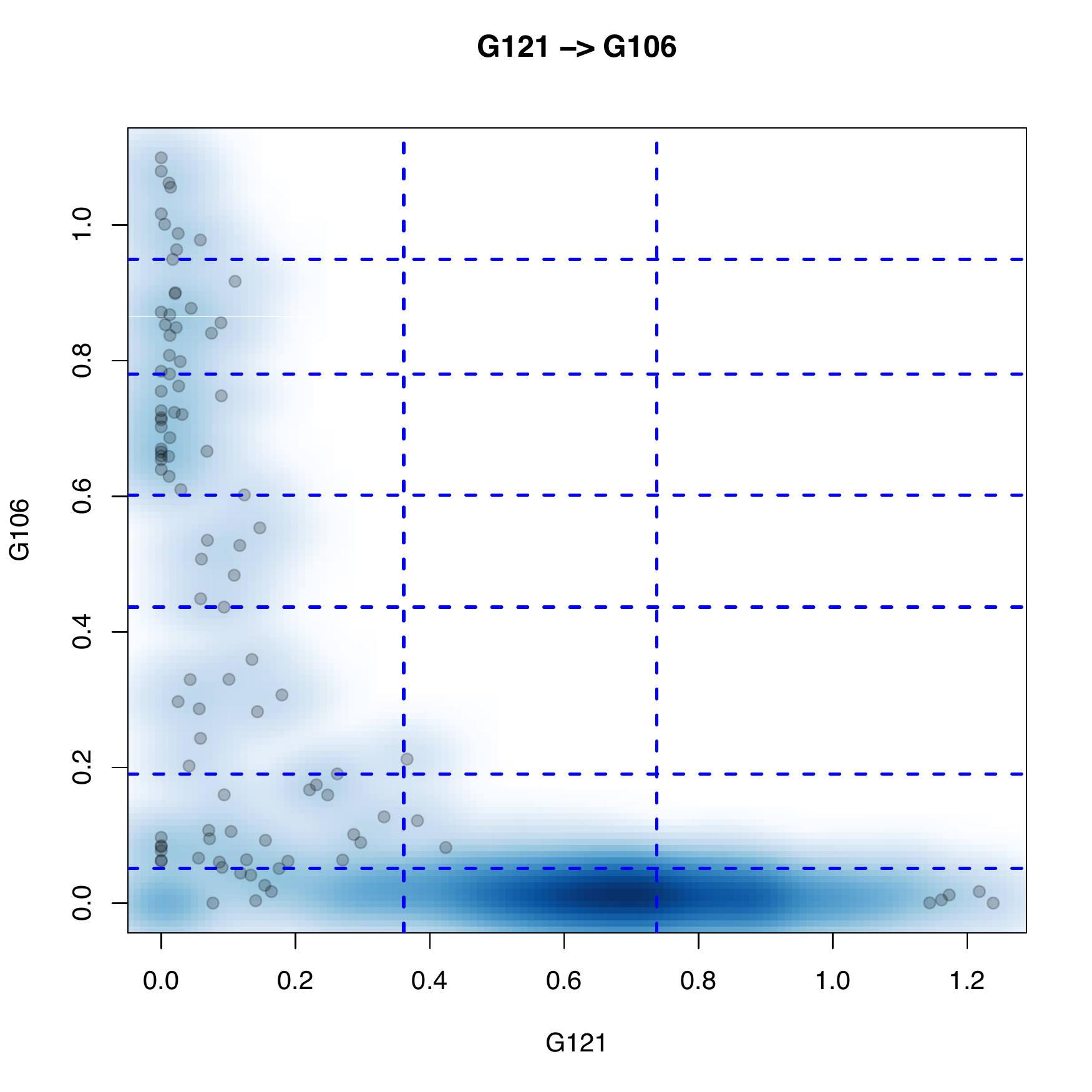}\\
\textbf{\textsf{c}} \; \includegraphics[width=.45\linewidth]{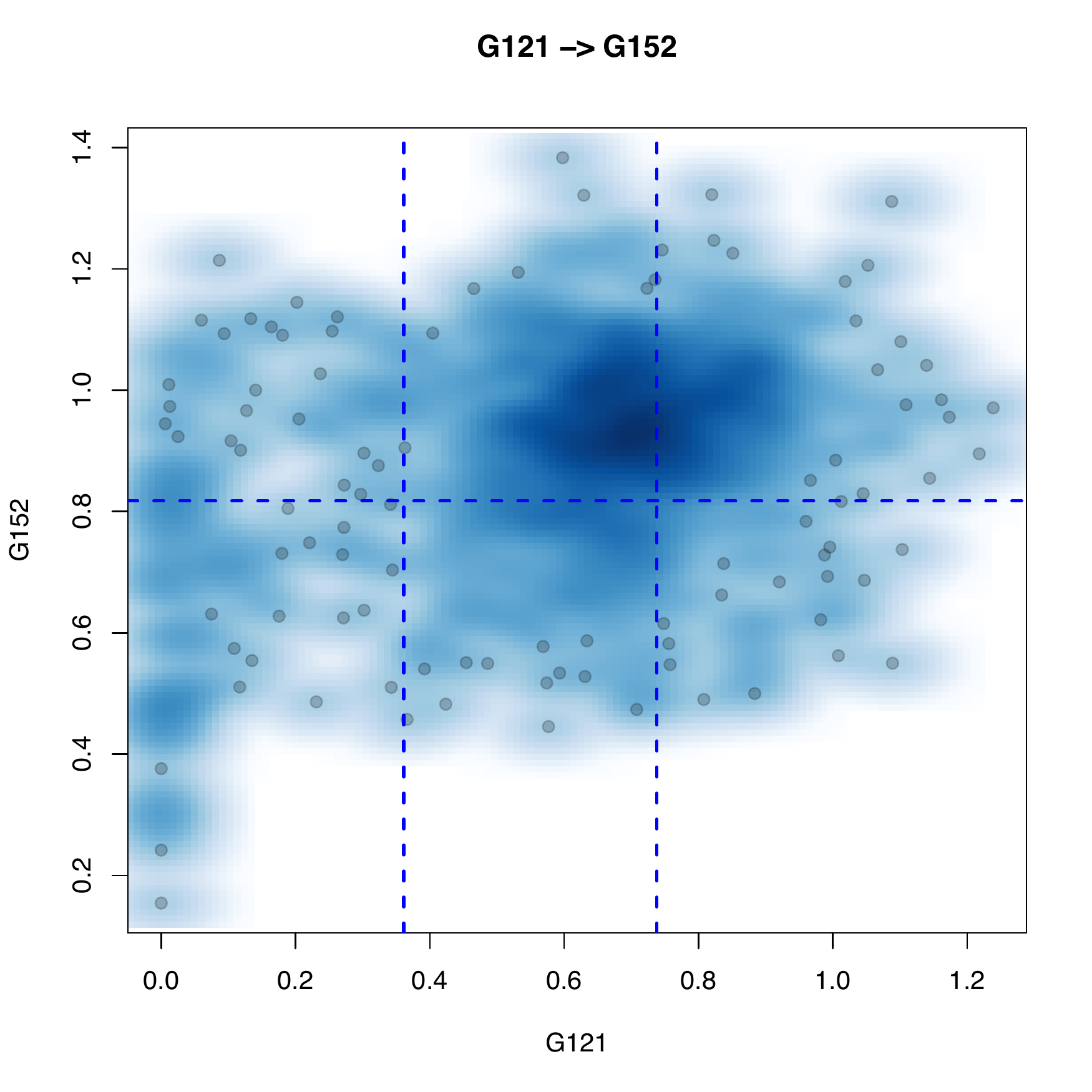}
\textbf{\textsf{d}} \; \includegraphics[width=.45\linewidth]{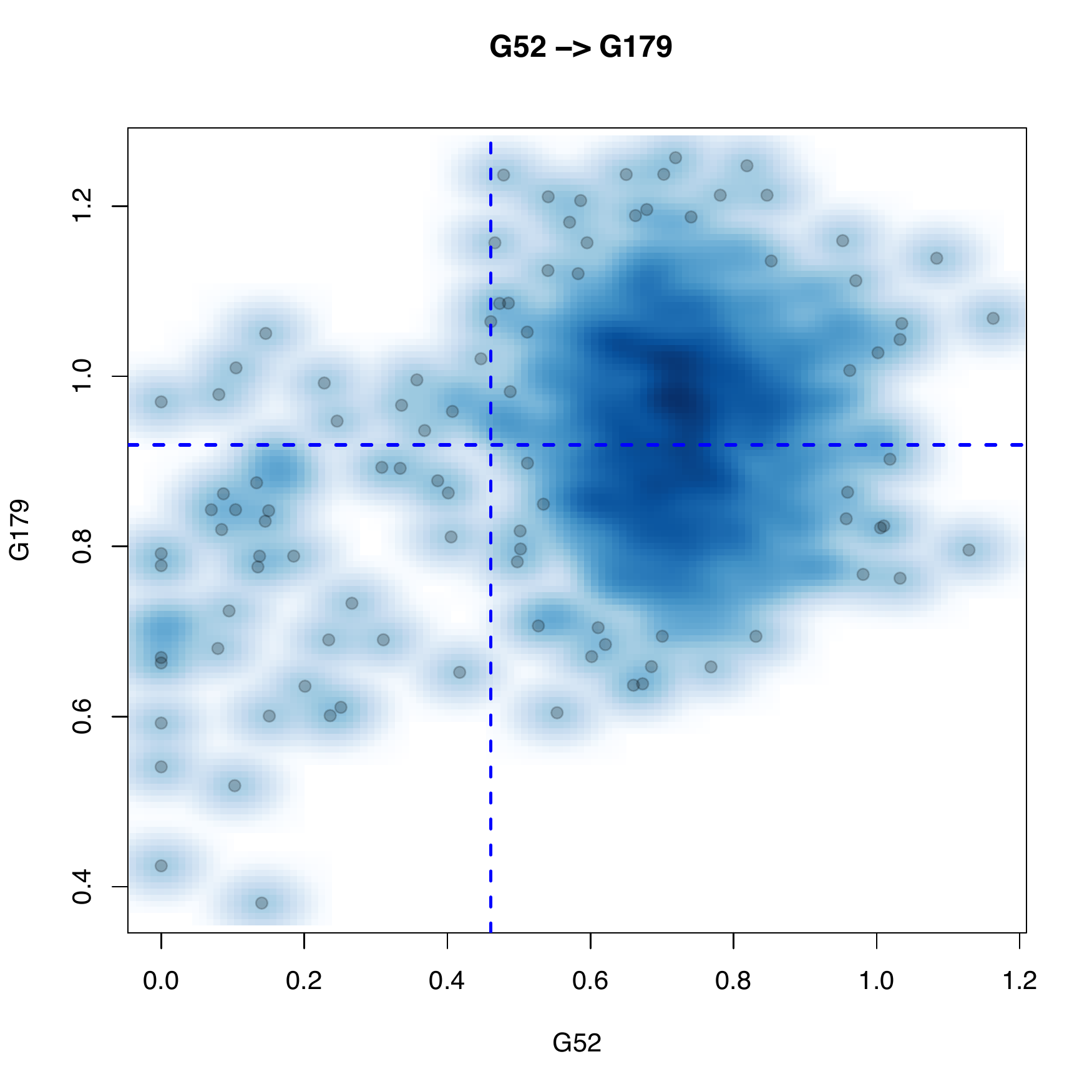}
\caption{{\bf Detected true interactions from in silico data used in DREAM5 Challenges demonstrate the advantage of \textsc{FunChisq}.}  Blue dashed lines are bin boundaries for quantization.  The horizontal axis represents $X$ and the vertical axis is for $Y$.  In these examples $Y$ exhibits much stronger nonlinear functional dependency on $X$ than the opposite.  They are representatives of typical nonlinear responses of $Y$ as $X$ increases: \textbf{(a)} Gradually decreasing response, \textbf{(b)} Sharply decreasing response, \textbf{(c)} Saturated response, and \textbf{(d)} Gradually increasing response.}
\label{fig:D5N1}
\end{figure}
Table~\ref{fig:D5N1} gives the functional chi-square statistics of those interactions in Fig.~\ref{fig:D5N1} as well as the ranking among all interactions we inspected for the data set.
\begin{table}[htbp]
\centering
\caption{The statistics of interactions demonstrating the effectiveness of \textsc{FunChisq} in Fig.~\ref{fig:D5N1}.}
\label{tab:D5N1}
\begin{tabular}{ccrccr}
\toprule
Interaction $X\to Y$ & Groundtruth & $\chi^2(f:X\to Y)$ & Degrees of freedom & $p$-value & Rank \\ 
\midrule
G121$\to$G131 & True & 248.08 & 2 & 1.35e-54 & 8\\
G131$\to$G121 & False & 179.10 & 2 & 1.28e-39 & 39\\
\midrule
G121$\to$G106 & True 	& 363.39 	& 12 &  2.09e-70 & 82 \\
G106$\to$G121 & False 	& 228.74 	& 12 & 3.64e-42 & 211  \\
\midrule
G121$\to$G152 & True &  55.09   &  2 & 1.09e-12   & 656 \\
G152$\to$G121 & False & 28.48    & 2  & 6.53e-07   & 3193 \\
\midrule
G52$\to$G179 & True & 25.14    & 1  &  5.34e-07   & 1878   \\
G179$\to$G52 & False & 5.14    & 1  &  0.0234  &  17352   \\
\bottomrule
\end{tabular}
\end{table}

\subsection{Evaluation of \textsc{FunChisq} on DREAM5 Challenges E.\ coli and yeast microarray data set}

We also evaluated \textsc{FunChisq} on DREAM5 E.\ coli and yeast microarray data sets.  The performance of the four methods are statistically close and there was no clear winner.  We inspected the data sets and discovered that the perturbations applied in the experimental design generated mostly either linear or normally distributed scatter plots.  We ponder that the lack of nonlinear dynamics may explain why the \textsc{FunChisq} did not stand out on both data sets. Figure \ref{fig:D5N3N4} shows two true interactions demonstrating the effectiveness of \textsc{FunChisq} that we identified from microarray data of E.\ coli and yeast.
\begin{figure}[htbp]
\centering
\textbf{\textsf{a}} \; \includegraphics[width=.45\linewidth]{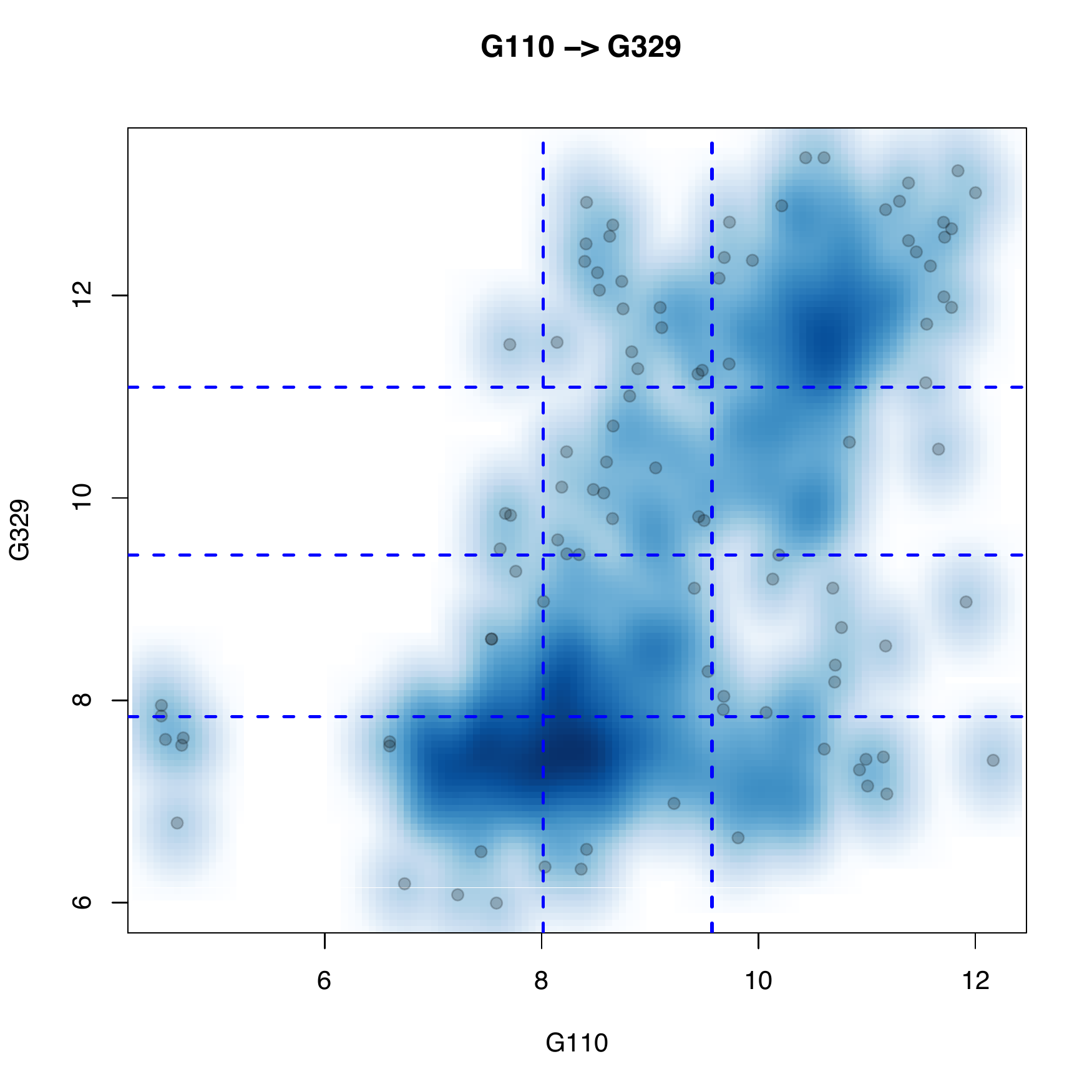}
\textbf{\textsf{b}} \; \includegraphics[width=.45\linewidth]{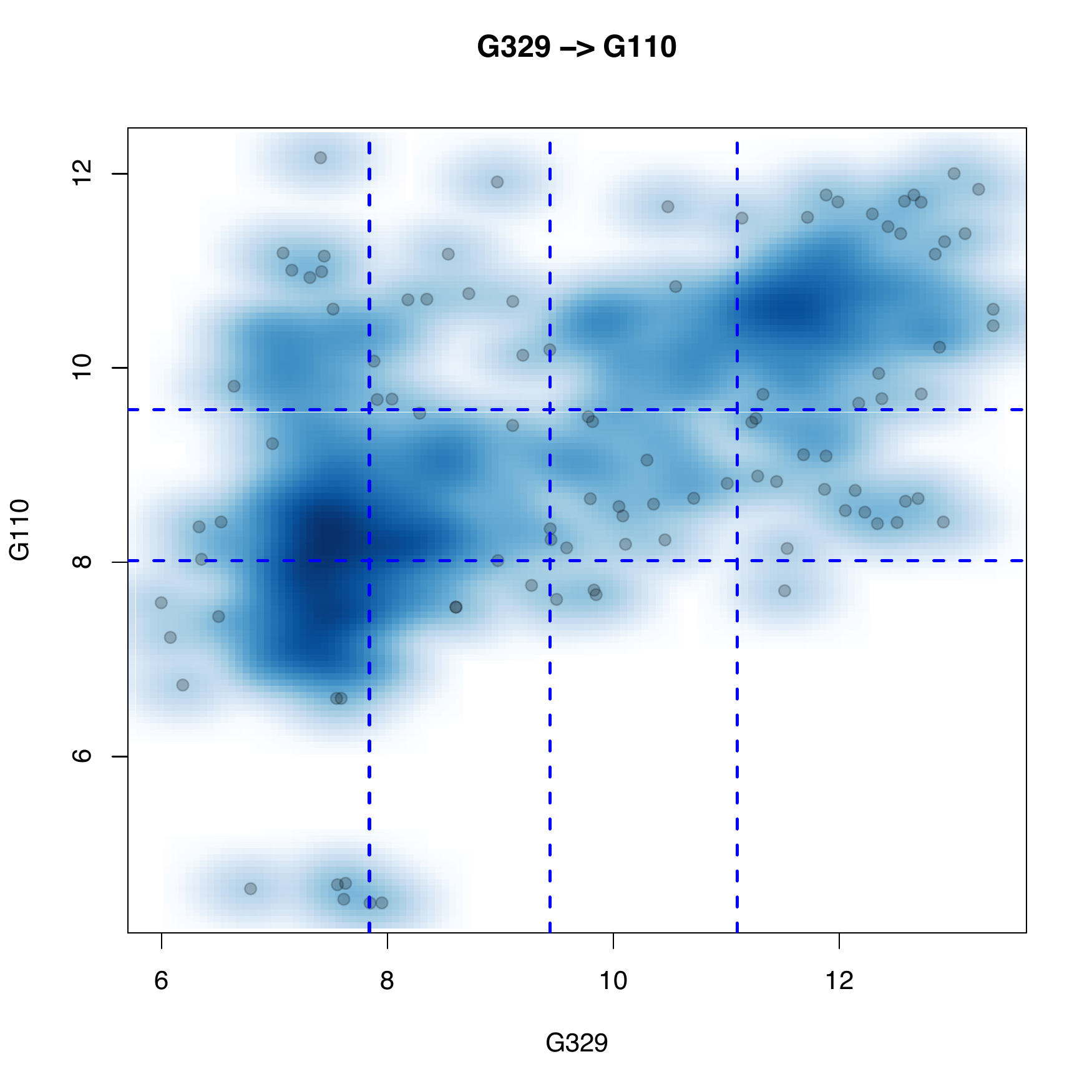}\\
\textbf{\textsf{c}} \; \includegraphics[width=.45\linewidth]{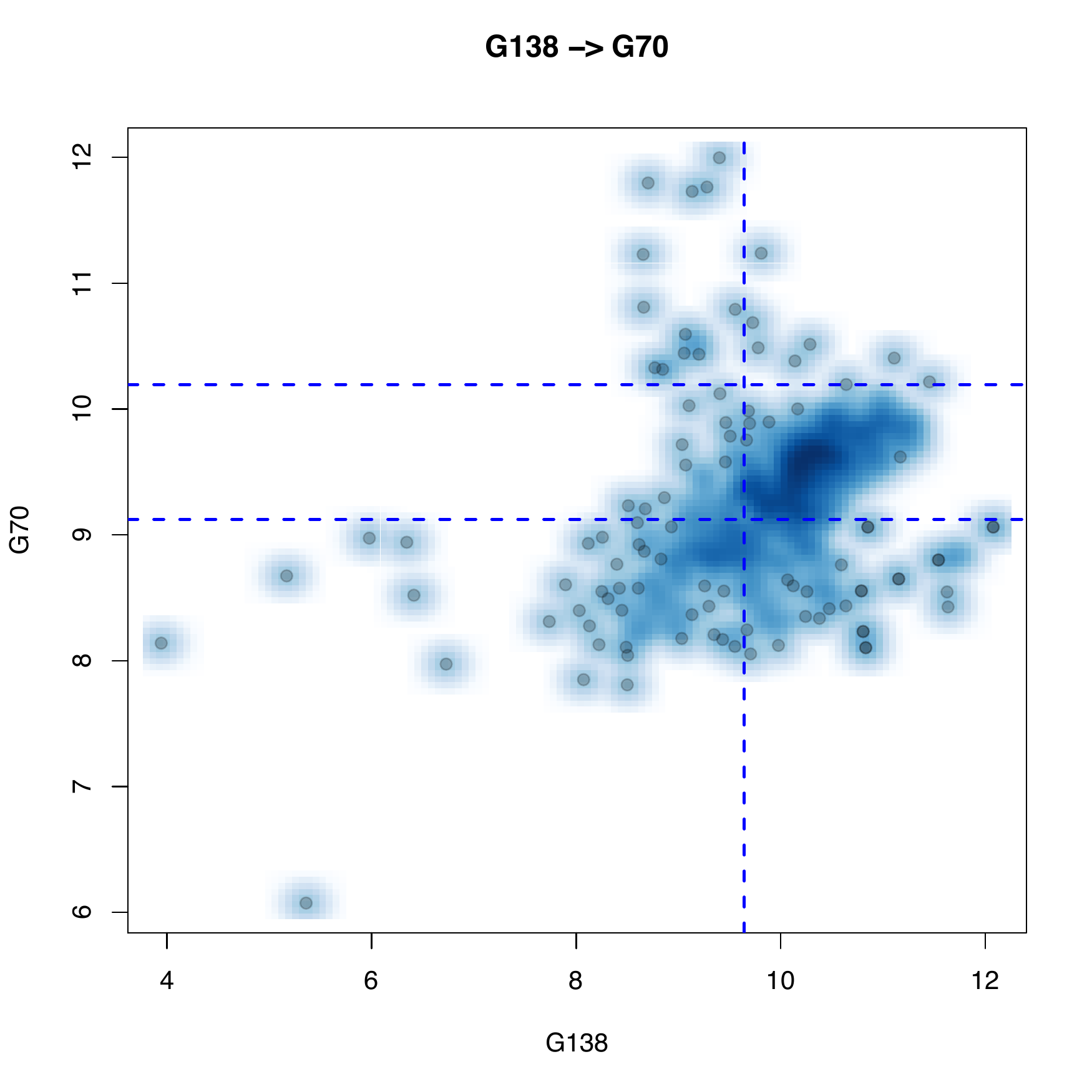}
\textbf{\textsf{d}} \; \includegraphics[width=.45\linewidth]{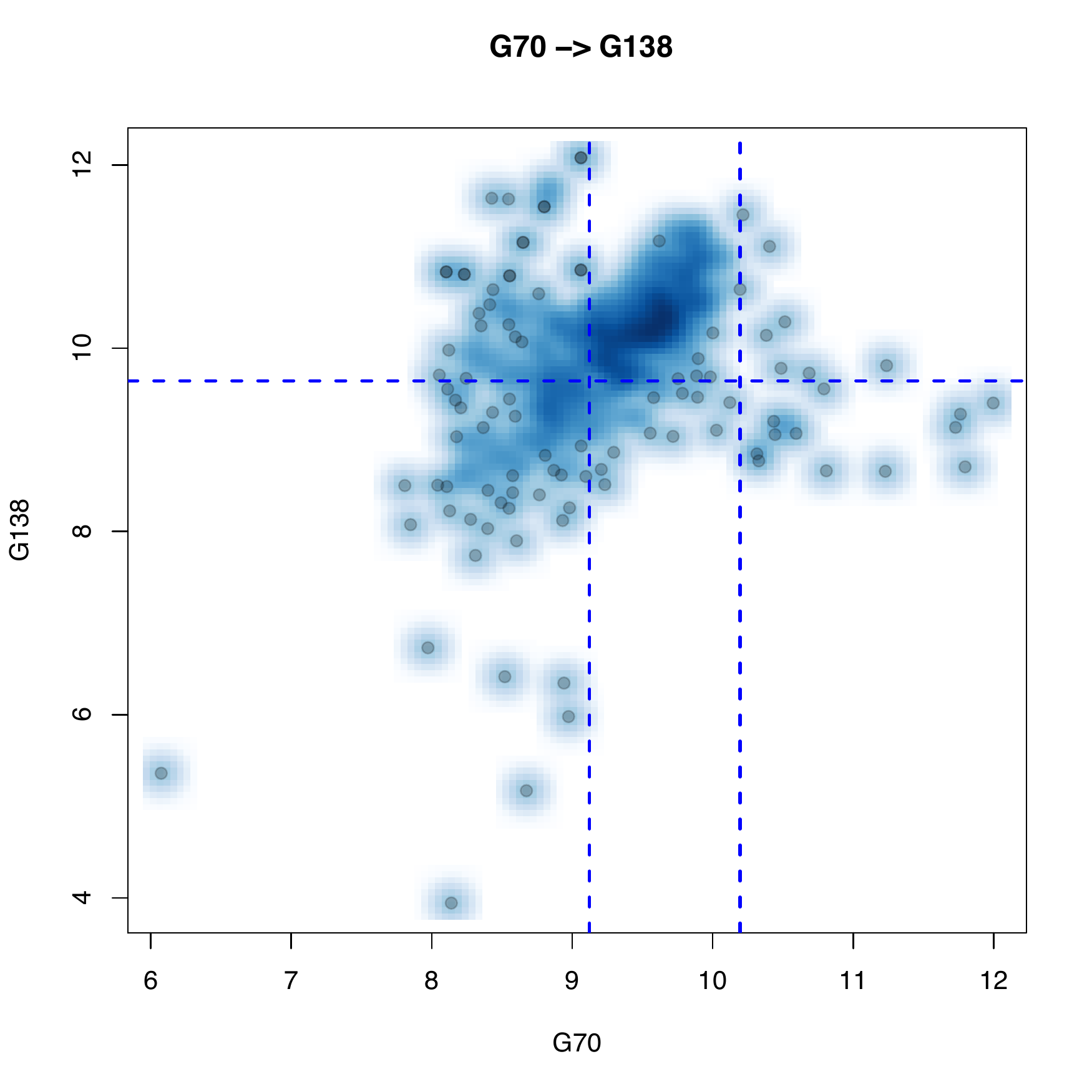}
\caption{{\bf Detected true interactions from DREAM5 Challenges E.\ coli and yeast microarray gene expression data demonstrate the advantage of \textsc{FunChisq}.}  Blue dashed lines are bin boundaries for quantization.  The horizontal axis represents $X$ and the vertical axis is for $Y$.  \textbf{(a)} A true interaction.   \textbf{(b)} The reverse of (a), an incorrect interaction.  \textbf{(c)} A second true interaction.  \textbf{(d)}  The reverse of (c), an incorrect interaction.}
\label{fig:D5N3N4}
\end{figure}
Table~\ref{fig:D5N3N4} gives the functional chi-square statistics of those interactions in Fig.~\ref{fig:D5N3N4} as well as the ranking among all interactions we inspected for the data set.  These are the best examples that we can identify from the microarray data sets.  The observed interaction pattern suggests potential directions for refining future experimental design.  
\begin{table}[htbp]
\centering
\caption{The statistics of interactions demonstrating the effectiveness of \textsc{FunChisq} in Fig.~\ref{fig:D5N3N4}.}
\label{tab:D5N3N4}
\begin{tabular}{ccrccr}
\toprule
Interaction $X\to Y$ & Groundtruth & $\chi^2(f:X\to Y)$ & Degrees of freedom & $p$-value & Rank \\ 
\midrule
G110$\to$G329 & True & 509.76 & 6 & 6.65e-107 & 138 \\
G329$\to$G110 & False & 451.07 & 6 & 2.89e-94& 253\\
\midrule
G138$\to$G70 & True 	&  184.72 	& 2 & 7.73e-41 & 780 \\
G70$\to$G138 & False 	&  118.18	& 2 & 2.18e-26 & 4253\\
\bottomrule
\end{tabular}
\end{table}

\section{Discussion}

\subsection{Related work to directionality test}

Conditional entropy $H(Y|X)=H(X,Y)-H(X)$ has been used in inferring gene networks \citep{Zhao2006inferring,Lopes2011}.  It measures how much uncertainty in $Y$ cannot be explained by $X$. A strong association between $X$ and $Y$ minimizes $H(Y|X)$.  Figure~\ref{fig:condentropy} provides a counter example that conditional entropy may overlook functional dependencies between two discrete variables.  There is no obvious way to fix this general flaw.
\begin{figure}[htbp]
\centering
\textbf{\textsf{a}} \quad
\begin{tabular}{ccc}
\toprule
     	& \#$(Y=1)$ 	& \#$(Y=2)$ \\
\midrule
$X=1$ &  0 	& 5 \\
$X=2$ &  0 	& 5 \\
\bottomrule
\end{tabular}
\hspace{.5in}
\textbf{\textsf{b}} \quad
\begin{tabular}{ccc}
\toprule
     	& \#$(Y=1)$ 	& \#$(Y=2)$ \\
\midrule
$X=1$ &  1 	& 4 \\
$X=2$ &  4 	& 1 \\
\bottomrule
\end{tabular}
\caption{{\bf Conditional entropy favors constant over non-constant functions.}  \textbf{(a)} The table has $H(Y|X)=0$, but the underlying constant function is not interesting. \textbf{(b)} The table has $H(Y|X)>0$ but $Y$ is linearly dependent on $X$.}
\label{fig:condentropy}
\end{figure}

ANOVA uses a mean-variance representation of each discrete group of $X$ based on normality assumption about noise \citep{Kuffner2012}.  It assumes that response continuous variable $Y$ is a linear function of the group means.  This mean-variance representation does not distinguish whether multiple or single peak are associated with the mean and variance.  This may thus make ANOVA insensitive to complex nonlinear functional relationships. 

Similar to ANOVA, logic regression \citep{Ruczinski2003logic} is related to \textsc{FunChisq} in that it also establishes functional dependency of a continuous variable on multiple discrete variables.  The difference is an additive parametric form is assumed for the underlying functional relationship, while \textsc{FunChisq} assumes no parametric forms.

Network deconvolution \citep{Feizi2013,Barzel2013} is designed to eliminate indirect interactions from direct ones, based on linear additivity of information flow in networks.  Our work is complementary from the nonlinearity angle and assumes functionality instead of linear additivity. 

\subsection{Generalization to testing against non-uniform response distributions}

We can generalize the functional chi-square test to a response chi-square test against a null response (column) distribution $p^0=[p_1^0,\ldots,p_s^0]^{\top}$ other than the uniform distribution as follows: 
\begin{equation}
\chi^2(p^0) = \sum_{i=1}^r \sum_{j=1}^s \frac{(n_{ij}-n_{i\cdot} p_j^0)^2}{n_{i\cdot} p_j^0} - \sum_{j=1}^s \frac{(n_{\cdot j}-np_j^0)^2}{n p_j^0}
\end{equation}
which measures deviation from the response distribution contributed to by the row variable.  When $p_j^0=n_{\cdot j}/{n}$, $\chi^2(p^0)$ is Pearson's chi-square, measuring departure from column marginal distribution contributed to by the row variable.  When $p_j^0=1/s$, $\chi^2(p^0)$ is the functional chi-square, measuring departure from uniform distribution contributed to by the row variable.  When one wishes to test how the response distribution have been modified from the one observed at an unperturbed condition, $p^0$ can be set to the observed conditional response distribution given the unperturbed condition.

\subsection{Implications to biological experimental design}

The functional chi-square test works most effectively when the underlying functional relationship is nonlinear.  This fits well to biological systems where nonlinear functional relationships are deeply rooted at both small and large scales, for example, the Hill equation for ligand and macromolecule binding.  In addition to design experiments around EC1 (under-saturated), EC50 (linear working zone), and EC99 (saturated), it is desirable for \textsc{FunChisq} to have data for responses at EC25 and EC75 to capture nonlinear-linear-nonlinear transition.  The effectiveness of \textsc{FunChisq} with proteomics data sets (not shown) using recent biotechnology indicates a potential convergence towards nonlinear dynamics between experimental design and network inference.  Therefore, contrary to perturbation theory where small perturbations are applied to find a solution by linear approximation, our results suggest that large perturbations can facilitate network dependency structure identification for a nonlinear system.

\subsection{Future work}

It is possible to use \textsc{FunChisq} in the context of temporal and functional dependencies to provide stronger evidence for causality.  For systems with complex dynamics, for example, a bistable switch, we can instead model the rate of change as a function of the stimulus to avoid non-functional dynamics between the variables.  Most importantly, \textsc{FunChisq} can be used to understand propagation of information in other types of dependency network, such as those in genome-wide association studies.

\bibliographystyle{apalike}

\bibliography{FunChisq.bib}

\appendix

\section{Zeros of the functional chi-square statistic} 

\label{AppZeros}

\begin{proposition}
$\chi^2(f:X\to Y)$ is zero if the empirical joint distribution of $X$ and $Y$ can be factorized as $\hat{P}(X,Y)=\hat{P}(X)\hat{P}(Y)$, or $X$ and $Y$ are empirically statistically independent $\hat{P}(Y|X)=\hat{P}(Y)$.
\label{prop:zeros}
\end{proposition}
\begin{proof}
This can be readily proved by the definition of the functional chi-square statistic. By empirically $\hat{P}(X,Y)=\hat{P}(X)\hat{P}(Y)$, we mean
\begin{equation}
\hat{P}(X,Y)=\frac{n_{ij}}{n} = \hat{P}(X)\hat{P}(Y) = \frac{n_{i\cdot}}{n} \cdot \frac{n_{\cdot j}}{n} \quad \text{or} \quad n_{ij}=\frac{n_{i\cdot} \cdot n_{\cdot j}}{n}
\end{equation}
Plugging $n_{ij}$ into Eq.~(\ref{eq:FCdef}), we obtain 
\begin{align}
\chi^2(f:X\to Y) 
 & = \sum_{i=1}^r \sum_{j=1}^s \frac{\left(\frac{n_{i\cdot} \cdot n_{\cdot j}}{n}-n_{i\cdot}/{s}\right)^2}{{n_{i\cdot}}/{s}} - \sum_{j=1}^s \frac{(n_{\cdot j}-n/s)^2}{n/s} \\
 & = \sum_{i=1}^r \sum_{j=1}^s \frac{n_{i\cdot}}{n} \frac{(n_{\cdot j}-n/s)^2}{n/s} - \sum_{j=1}^s \frac{(n_{\cdot j}-n/s)^2}{n/s} \\
 & = 0
\end{align}

\end{proof}

\section{The null distribution of the functional chi-square test}

\label{AppNull}

\def \iit {\mathbf{11}^{\top}} 
\begin{lemma}
The chi-square statistic in the goodness-of-fit test with $s$ classes can be decomposed into the sum of $s-1$ independent chi-squares of 1 degree of freedom when the expected frequencies are the same for each class.
\label{l:chisqde}
\end{lemma}
The proof of Lemma~\ref{l:chisqde} can be found in \citep{Boero2004Decompositions}.  Here we summarize important results used in the proof before we apply them to the functional chi-square statistic.  The observed frequencies $n_1,\ldots,n_s$ in each of $s$ classes after $n$ trials generally follow a multinomial distribution with equal success probability $p_j=1/s$ $(j=1,\ldots,s)$.  The covariance matrix is given by
\[
	\Sigma = (n/s)(I - \iit/s) 
\]
where $\mathbf{1}$ is a vector of $s$ 1's.  Since matrix $I - \iit/s$ of rank $s-1$ is both idempotent and Hermitian, it can be factorized as \citep{Rao1998Matrix}
\begin{equation}
I - \iit/s = V^{\top}V 
\label{eq:V}
\end{equation}
where $V$ is a $(s-1)\times s$ matrix satisfying $VV^{\top}=I$.  The rows of $V$ are the $s-1$ eigenvectors associated with the non-zero eigenvalues of matrix $I - \iit/s$.  Therefore we have
\[
\Sigma = (n/s) V^{\top}V
\]
Let the standardized frequency vector be
\[
a = [a_1,\ldots,a_s]^{\top} = \left[\frac{n_1-n/s}{\sqrt{n/s}},\ldots,\frac{n_s-n/s}{\sqrt{n/s}}\right]^{\top}
\]
The transformation $e = Va$ identifies $s-1$ independent and asymptotically standard normal variables in vector $e=[e_1,\ldots,e_m]$ such that \citep{Boero2004Decompositions}
\begin{equation}
\chi^2 = \sum_{j=1}^s \frac{(n_j-np_j)^2}{np_j} = \|Va\|^2 = \sum_{m=1}^{s-1} e_m^2 
\end{equation}
It immediately follows that $e_1^2,\ldots,e_{s-1}^2$ are $s-1$ independent chi-square variables of 1 degree of freedom. 

\begin{proposition}
Each conditional row (parent) chi-square $\chi^2(Y|X=i)$ can be decomposed to a sum of $s-1$ independent chi-squares with 1 degree of freedom.
\end{proposition}

\begin{proof}
We apply chi-square decomposition in Lemma~\ref{l:chisqde} on the conditional chi-square statistic of row $i$ defined by 
\[
	\chi^2(Y|X=i) = \sum_{j=1}^s \frac{(n_{ij}-n_{i\cdot}/s)^2}{n_{i\cdot}/s}
\]
Let the standardized frequency in each of the $k$ cell in row $i$ be
\begin{equation}
	a_{ij} = \frac{n_{ij}-n_{i\cdot}/s}{\sqrt{n_{i\cdot}/s}}
\label{eq:stdfrqi}
\end{equation}
By transforming with $(s-1) \times s$ matrix $V$ (Eq.~\ref{eq:V})
\[
[e_{i,1}, \ldots, e_{i,s-1}]^{\top} = V [a_{i1}, \ldots, a_{is}]^{\top}
\]
into $s-1$ independent standard normal variables, we obtain the chi-square decomposition
\begin{equation}
	\chi^2(Y|X=i) = e_{i,1}^2 + \cdots + e_{i,s-1}^2
\end{equation}
of $s-1$ components of independent $\chi^2_1$ each with 1 degree of freedom.
\end{proof}

\begin{proposition}
The column (child) marginal chi-square $\chi^2(Y)$ can be decomposed into the sum of $s-1$ independent chi-square random variables with 1 degree of freedom.
\end{proposition}

\begin{proof}
By Lemma~\ref{l:chisqde}, we can represent the child marginal chi-square $\chi^2(Y)$ 
\[
\chi^2(Y)=\sum_{j=1}^s \frac{(n_{\cdot j}-n/s)^2}{n/s}
\]
using the chi-square components
\[
	\chi^2(Y) =  e_{1}^2 + \cdots + e_{s-1}^2
\]
where the $s-1$ dimension vector of independent standard normal variables  
\[
[e_{1}, \ldots, e_{s-1}]^{\top} = V [a_{1}, \ldots, a_{s}]^{\top}
\]
transformed from $s$ standardized marginal frequencies for the child variable defined by   
\begin{equation}
a_{j} = \frac{n_{\cdot j}-n/s}{\sqrt{n/s}} = \sum_{i=1}^r \sqrt{\frac{n_{i\cdot}}{n}} a_{ij}
\label{eq:aj}
\end{equation}
In Eq.~(\ref{eq:aj}), the first equality is by definition and the second by Eq.~(\ref{eq:stdfrqi}).
\end{proof}

\begin{theorem}
Under the null hypothesis that discrete random variable $Y$ is uniformly distributed conditioned on random variable $X$, $\chi^2(f:X\to Y)$ follows a chi-square distribution asymptotically.
\label{funchisq}
\end{theorem}

\begin{proof}
Although the difference between two independent chi-square random variables can be negative and is thus no longer chi-squared, we show that the functional chi-square statistic, the sum of row conditional chi-squares subtracted by the column chi-square, indeed gives rise to another chi-square random variable.

\def \vm {\left[v_{m1} \ldots v_{ms}\right]}
\def \eemT {\left[e_{1m} \ldots e_{rm}\right]}
\def \eem {\left[ \begin{array}{c} e_{1m} \\ \vdots \\ e_{rm} \end{array} \right]}

\def \ninT {\left[ \sqrt{\frac{n_{1\cdot}}{n}} \cdots \sqrt{\frac{n_{r\cdot}}{n}} \right]}
\def \nin {\left[ \begin{array}{c} \sqrt{\frac{n_{1\cdot}}{n}} \\ \vdots \\ \sqrt{\frac{n_{r\cdot}}{n}} \end{array} \right]}

Representing $V=[v_{mi}]$, defined in Eq.~(\ref{eq:V}), and for each conditional row chi-square $\chi^2(Y|X=i)$, we write its standard normal component as
\begin{equation}
e_{im} = \vm \left[ \begin{array}{c} a_{i1}\\ \vdots \\ a_{is} \end{array} \right]
\end{equation}
Similarly, we represent each standard normal component of the child marginal chi-square by  
\begin{align}
e_{m} &= \vm \left[ \begin{array}{c} a_{1}\\ \vdots \\ a_{s} \end{array} \right] \\
&= \vm \left[ \begin{array}{c} \sum_{i=1}^r \sqrt{\frac{n_{i\cdot}}{n}} a_{i1} \\ \vdots \\ \sum_{i=1}^r \sqrt{\frac{n_{i\cdot}}{n}} a_{is} \end{array} \right] \\
&= \vm \left[ \begin{array}{ccc} a_{11} & \cdots & a_{r1} \\ 
								\vdots & \ddots & \vdots \\ 
								a_{1s} & \cdots & a_{rs}  
	 \end{array} \right]
     \nin\\
&= \eemT \nin
\end{align}

Plugging in the above variables, the functional chi-square can be re-written as
\begin{align}
\chi^2 =& \left[\sum_{i=1}^r \chi^2(Y|X=i)\right] - \chi^2(Y)
= \left( \sum_{i=1}^r  \sum_{m=1}^{s-1} e_{im}^2 \right) - \sum_{m=1}^{s-1} e_{m}^2 
= \sum_{m=1}^{s-1} \left( \sum_{i=1}^r e_{im}^2 - e_{m}^2 \right) \\
=& \sum_{m=1}^{s-1} \left\{ \eemT \eem - \eemT \nin \ninT \eem \right\} \\
=& \sum_{m=1}^{s-1} \eemT \left\{ \underbrace{I - \nin \ninT}_{\text{Matrix } O \text{: idempotent with rank } r-1} \right\} \eem \label{eq:decomp}
\end{align}
Since matrix $O$ is idempotent with a rank of $r-1$ and $\eemT^{\top}$ is a vector of independent standard normal variables, the quadratic form $\sum_{i=1}^r e_{im}^2 - e_{m}^2$ is chi-squared with $r-1$ degrees of freedom \citep{Mathai1992}.  As $\eemT^{\top}$ over different $m$ are independent vectors, the above chi-squares for different $m$ are independent.  As the summation of independent chi-squares are still chi-squared with the degrees of freedom summed \citep{Casella2002}, it follows immediately that $\chi^2$ is asymptotically chi-square distributed with $(s-1)(r-1)$ degrees of freedom under the null hypothesis of no functional dependency. 
\end{proof}

\begin{corollary}
The functional chi-square statistic $\chi^2(f:X\to Y)$ is non-negative. 
\label{cor:nn}
\end{corollary}
\begin{proof}
In the decomposition in Eq.~(\ref{eq:decomp}), matrix $O$ is idempotent and symmetric, that is
\begin{equation}
O^{\top}O=OO=O
\end{equation} 
Therefore, for any vector $x$, we have
\[
	x^{\top}Ox = x^{\top}O^{\top}Ox = (Ox)^{\top} Ox = \left\| Ox \right\|^2 \ge 0
\]
which implies that matrix $O$ is positive semi-definite.  Thus the $s-1$ quadratic forms involving $O$ in Eq.~(\ref{eq:decomp}) are always non-negative.  The sum of non-negative terms leads to non-negativity of $\chi^2$.  This is true mathematically (including asymptotically). 
\end{proof}

\section{Optimality of the functional chi-square test}

\label{AppOpt}

\begin{proposition}
Given $s$ non-negative numbers $x_1,\ldots,x_s$ that sum up to a constant $a\ge 0$, and another constant $c$, it follows that 
\[
\sum_{j=1}^s (x_j - c)^2 \le (a - c)^2 + (s-1)c^2
\]
where the equality holds true if and only if $x_m=a$ for some unique $m\in\{1,\ldots,s\}$ and all other $x_j$'s are zero.
\label{l:ub}
\end{proposition}
\begin{proof}
Starting from the left hand side of the inequality, we can derive
\begin{align*}
   \sum_{j=1}^s (x_j - c)^2
& = sc^2 - 2ac + \sum_{j=1}^s x_j^2 \\
& = sc^2 - 2ac + a^2 - 2\sum_{i=1}^s \sum_{j=i+1}^s x_i x_j\\
& = (a-c)^2 + (s-1) c^2 - 2\sum_{i=1}^s \sum_{j=i+1}^s x_i x_j\\
& \le (a-c)^2 + (s-1) c^2  \quad  (\because x_i, x_j \ge 0)
\end{align*}
The equality will hold in this last equality above if and only if $x_j=0$ for all $j$ except $x_m=a$ for some $m\in\{1,\ldots,s\}$. 
\end{proof}

\begin{theorem}
A contingency table of sample size $n$ and observed marginal distribution of column variable $Y$ of $s$ levels, $q=[q_1,\ldots,q_s]$, maximizes $\chi^2(f:X\to Y)$ if and only if $Y$ is a function of the row variable $X$ when such a contingency table is feasible.  The upper bound to the functional chi-square is given by $ns \cdot \left(1-\sum_{j=1}^s q_j^2\right)$.
\label{Opt}
\end{theorem}

\begin{proof}
Applying Proposition~\ref{l:ub} to the row conditional chi-square statistic, we obtain 
\begin{align}
\chi^2(Y|X=i) &= \sum_{j=1}^s \frac{(n_{ij}-n_{i\cdot}/{s})^2}{{n_{i\cdot}}/{s}} \\
& \le \frac{(n_{i\cdot}-n_{i\cdot}/{s})^2 + (s-1)(n_{i\cdot}/{s})^2}{{n_{i\cdot}}/{s}}\\
& = n_{i\cdot} (s-1)
\end{align}
where the equality holds true if and only if there is one non-zero entry in the row if the table does not violate the row and column sums.  When such a single non-zero entry exists for every row, it implies that $Y$ is a function of $X$.

Since $q=[q_1,\ldots,q_s]$ is the observed marginal distribution of $Y$, we have $q_j=n_{\cdot j}/n$.  Plugging the upper bound of each row chi-square to the definition of functional chi-square, we obtain 
\begin{align}
\chi^2(f:X\to Y) &= \sum_{i=1}^r \chi^2(Y|X=i) - \chi^2(Y) \\
& \le n(s-1) - \sum_{j=1}^s \frac{(n_{\cdot j}-n/s)^2}{n/s}\\
& = ns\left(1-\sum_{j=1}^s q_j^2\right)
\end{align}
which bounds the functional chi-square statistic from above.
\end{proof}

\section{Using \textsc{FunChisq} software}

\label{AppUsage}

\subsection{Discretization}

As \textsc{FunChisq} uses a discrete nonparametric representation for interactions, continuous data must be quantized first.  The guideline for discretization is to preserve qualitative trends in the data but eliminate noisy fluctuations.  Here we describe a procedure that we used in analyzing DREAM5 data sets.

We first use R package \texttt{mclust} \citep{Fraley2003} to determine for each variable the number of quantization levels needed.  The only parameter needed here is the maximum possible number of peaks of each continuous random variable.  A program in the \texttt{mclust} package estimates the actual number of peaks $k$ specific to each variable using a Gaussian mixture model that optimizes the Bayesian information criterion.

We further discretize each variable independently using a $k$-means method -- exact and optimal for a single variable -- that we have developed and implemented as R package Ckmeans.1d.dp \citep{Wang2011}.

\subsection{Implementation of the functional chi-square test in R}

An R implementation of the functional chi-square test is given below as function \texttt{fun.chisq.test()}.  The input is a matrix of nonnegative values representing a contingency table $x$ and the output is a list of the functional chi-square statistic, the degrees of freedom, and the $p$-value, associated with the given contingency table $x$. 

\begin{Verbatim}[numbers=left,label=The functional chi-square test,frame=lines]
fun.chisq.test <- function (x)
{
  row.chisq.sum <- sum(apply(x, 1, 
                             function(v){ 
                               if(sum(v)>0) chisq.test(v)$statistic
                               else 0
                               } 
                             )
                       )
  
  fun.chisq <- row.chisq.sum - chisq.test( apply(x, 2, sum) )$statistic

  df <- nrow(x) * (ncol(x) - 1) - (ncol(x) - 1)
  p.value <- pchisq(fun.chisq, df = df, lower.tail=FALSE)

  return( list( statistic=fun.chisq, parameter=df, p.value=p.value ) )
}
\end{Verbatim}

The following R code reproduces the example used in Fig.~\ref{fig:ex}, by calling \texttt{fun.chisq.test()} and also the Pearson's chi-square test function \texttt{chisq.test()}.

\begin{Verbatim}[frame=lines,label=Example,numbers=left]
w <- matrix(c(5,1,5,1,5,1,1,0,1), nrow=3)
u <- t(w)

w.fun <- fun.chisq.test(w)
u.fun <- fun.chisq.test(u)

cat("Functional chisq(w):\t", w.fun$statistic, "\t", w.fun$parameter, "\t", 
	w.fun$p.value, "\n")
cat("Functional chisq(u):\t", u.fun$statistic, "\t", u.fun$parameter, "\t", 
	u.fun$p.value, "\n")

w.pearson <- chisq.test(w)
cat("Pearson chisq(w or u):\t", w.pearson$statistic, "\t", 
	w.pearson$parameter, "\t", w.pearson$p.value, "\n")
\end{Verbatim}

The output on the screen as generated by the above example code is given as follows:
\begin{Verbatim}[frame=lines,label=Output]
Functional chisq(w):	 10.04286 	 4 	 0.03971191 
Functional chisq(u):	 8.380519 	 4 	 0.07859274 
Pearson chisq(w or u):	 8.868275 	 4 	 0.06447766 
\end{Verbatim}

\end{document}